\documentclass[a4paper,10pt]{article}
\usepackage{graphicx,url,amsmath,amssymb,color,subfigure,anysize,amsthm,algorithm,algpseudocode,palatino}
\marginsize{3cm}{3cm}{3cm}{3cm}

\newcommand{\eps}{\varepsilon}

\newcommand{\p}{\partial}
\newcommand{\mc}{\mathbf{c}}
\newcommand{\md}{\mathbf{d}}
\newcommand{\me}{\mathbf{e}}
\newcommand{\mn}{\mathbf{n}}

\newcommand{\mr}{\mathbf{r}}
\newcommand{\mt}{\mathbf{t}}

\newtheorem{thm}{Theorem}[section]

\newtheorem{lem}[thm]{Lemma}

\begin{document}
\title{Structure analysis of single- and multi-frequency subspace migrations in inverse scattering problems}
\author{Young Deuk Jo, Young Mi Kwon, Joo Young Huh, and Won-Kwang Park\thanks{Department of Mathematics, Kookmin University, Seoul, 136-702, Korea. parkwk@kookmin.ac.kr}}
\date{}
\maketitle
\begin{abstract}
  In this literature, we carefully investigate the structure of single- and multi-frequency imaging functions, that are usually employed in inverse scattering problems. Based on patterns of the singular vectors of the Multi-Static Response (MSR) matrix, we establish a relationship between imaging functions and the Bessel function. This relationship indicates certain properties of imaging functions and the reason behind enhancement in the imaging performance by multiple frequencies. Several numerical simulations with a large amount of noisy data are performed in order to support our investigation.
\end{abstract}

\section{Introduction}\label{sec:1}
One of the main objective of the inverse scattering problem is to identify the characteristics of unknown targets from measured scattered field or far-field pattern. In research fields such as physics, medical science, and materials engineering, this is an interesting and important problem. Related works can be found in \cite{A,C,DEKPS,D,DCGS} and references therein. In order to solve this problem, various algorithms for finding the locations and/or shapes of targets have been accordingly developed.

In many research studies \cite{ADIM,C,DEKPS,DL,LB,PL4}, the shape reconstruction method is based on Newton-type iterative algorithms. However, for a successful shape reconstruction using these algorithms, the iterative procedure must begin with a good initial guess that is close to the unknown target because it highly depends on the initial guess; for more details, refer to \cite{KSY,PL4}.

To finding a good initial guess, alternative non-iterative reconstruction algorithms have been developed, such at the MUltiple SIgnal Classification (MUSIC)-type algorithm \cite{AK,PL1,PL3}, linear sampling method \cite{CH,CHM}, topological derivative strategy \cite{AGJK,AKLP,MKP,P3,P4}, and the multi-frequency based algorithm such as Kirchhoff and subspace migrations \cite{A,AGKPS,P1,P2,PL2}. Among them, although the multi-frequency based subspace migration has exhibited potential as a non-iterative imaging technique, a mathematical identification of its structure needs to be performed for its heuristical applications, which is the motivation behind.

In this paper, by intensively analyzing the structure of single- and multi-frequency subspace migration, we discover some properties and confirm the reason behind the enhancement in the imaging performance by applying multiple frequencies. In recent work \cite{AGKPS}, this fact was verified by the Statistical Hypothesis Testing but our approach is to find a relationship between imaging functions and Bessel functions of the first kind of the integer order.

This paper is organized as follows. In section \ref{sec:2}, we briefly review the two-dimensional direct scattering problem, and an asymptotic expansion formula for far-field patterns, and introduce the imaging function introduced in \cite{P1}. In section \ref{sec:3}, we analyze the single- and multi-frequency based imaging functions and discuss their properties. In section \ref{sec:4}, we present several numerical experiments and discuss the effectiveness, robustness, and limitation of imaging functions. Finally, a brief conclusion is given in section \ref{sec:5}.

\section{Review on imaging function}\label{sec:2}
In this section, we survey the two-dimensional direct scattering problem and an imaging algorithm. A more detailed discussion can be found in \cite{AGKPS,AK,P1,P2,PL2}.

\subsection{Direct scattering problem and asymptotic expansion formula}
Let $\Sigma_m$ be a homogeneous inclusion with a small diameter $\rho$ in the two-dimensional space $\mathbb{R}^2$. Throughout this paper, we assume that every $\Sigma_m$ is expressed as
\[\Sigma_m=\mr_m+\rho\mathcal{D}_m,\]
where $\mr_m$ and $\rho$ denote the location and size of $\Sigma_m$, respectively. Here, $\mathcal{D}_m$ is a simple connected smooth domain containing the origin.

Let $\eps_0$ and $\mu_0$ respectively denote the dielectric permittivity and magnetic permeability of $\mathbb{R}^2$. Similarly, we let $\eps_m$ and $\mu_m$ be those of $\Sigma_m$. For simplicity, let $\Sigma$ be the collection of $\Sigma_m$, $m=1,2,\cdots,M,$ and we define the following piecewise constants:
\[\eps(\mr)=\left\{\begin{array}{ccl}
\eps_m&\mbox{for}&\mr\in\Sigma_m\\
\eps_0&\mbox{for}&\mr\in\mathbb{R}^2\backslash\overline{\Sigma}
\end{array}\right.
\quad\mbox{and}\quad
\mu(\mr)=\left\{\begin{array}{ccl}
\mu_m&\mbox{for}&\mr\in\Sigma_m\\
\mu_0&\mbox{for}&\mr\in\mathbb{R}^2\backslash\overline{\Sigma}.
\end{array}\right.\]
Throughout this paper, we assume that $\eps_0=\mu_0=1$ and $\eps_m>\eps_0$, $\mu_m>\mu_0$ for $m=1,2,\cdots,M$.
At a given frequency $\omega$, let $u_{\mbox{\tiny tot}}(\mr,\md_l;\omega)$ be the time-harmonic total field that satisfies the Helmholtz equation
\begin{equation}\label{HelmHoltzEquation}
  \nabla\cdot\bigg(\frac{1}{\mu(\mr)}\nabla u_{\mbox{\tiny tot}}(\mr,\md_l;\omega)\bigg)+\omega^2\eps(\mr)u_{\mbox{\tiny tot}}(\mr,\md_l;\omega)=0
\end{equation}
with transmission conditions at the boundaries of $\Sigma_m$.

Let $u_{\mbox{\tiny inc}}(\mr,\md_l;\omega)$ be the solution of (\ref{HelmHoltzEquation}) without $\Sigma$. In this paper, we consider the following plane-wave illumination: for a vector $\md_l\in\mathfrak{C}^1$, $u_{\mbox{\tiny inc}}(\mr,\md_l;\omega)=\exp(j\omega\md_l\cdot\mr)$. Here, $\mathfrak{C}^1$ denotes a two-dimensional unit circle.

Generally, the total field $u_{\mbox{\tiny tot}}$ can be divided into the incident field $u_{\mbox{\tiny inc}}$ and the unknown scattered field $u_{\mbox{\tiny scat}}$, which satisfies the Sommerfeld radiation condition
\[\lim_{|\mr|\to\infty}\sqrt{|\mr|}\left(\frac{\p u_{\mbox{\tiny scat}}(\mr,
\md_l;\omega)}{\p|\mr|}-j k_0u_{\mbox{\tiny scat}}(\mr,\md_l;\omega)\right)=0,\]
uniformly in all directions $\hat{\mr}=\frac{\mr}{|\mr|}$. Note that since we assumed $\eps_0=\mu_0=1$, wavenumber $k_0$ satisfies $k_0=\omega\sqrt{\eps_0\mu_0}=\omega$. As given in \cite{AK}, $u_{\mbox{\tiny scat}}$ can be written as the following asymptotic expansion formula in terms of $\rho$
\begin{multline}\label{ScatteredField}
  u_{\mbox{\tiny scat}}(\mr,\md_l;\omega)=\rho^2\sum_{m=1}^{M}\bigg(\nabla u_{\mbox{\tiny inc}}(\mr_m,\md_l;\omega)\cdot\mathbb{T}(\mr_m)\cdot\nabla\Phi(\mr_m,\mr,\omega)\\+\omega^2(\eps-\eps_0)\mbox{area}(\mathcal{D}_m)u_{\mbox{\tiny inc}}(\mr_m,\md_l;\omega)\Phi(\mr_m,\mr,\omega)\bigg)+o(\rho^2),
\end{multline}
where $o(\rho^2)$ is uniform in $\mr_m\in\Sigma_m$ and $\md_l\in\mathfrak{C}^1$. Here $\mbox{area}(\mathcal{D}_m)$ denotes the area of $\mathcal{D}_m$, $\mathbb{T}(\mr_m)$ is a $2\times2$ symmetric matrix:
\[\mathbb{T}(\mr_m)=\frac{2\mu_0}{\mu_m+\mu_0}\mbox{area}(\mathcal{D}_m)\mathbb{I}_2,\]
where $\mathbb{I}_n$ denotes the $n\times n$ identity matrix, and $\Phi(\mr_m,\mr,\omega)$ is the two-dimensional time harmonic Green function (or fundamental solution to Helmholtz equation)
\[\Phi(\mr_m,\mr,\omega)=-\mu_0\frac{j}{4}H_0^1(\omega|\mr_m-\mr|),\]
where $H_0^1$ is the Hankel function of order zero and of the first kind.

The far-field pattern is defined as function $F(\hat{\mr},\md_l)$ that satisfies
\begin{equation}\label{FarField}
  u_{\mbox{\tiny scat}}(\mr,\md_l;\omega)=\frac{\exp(jk_0|\mr|)}{\sqrt{|\mr|}}F(\hat{\mr},\md_l)+o\left(\frac{1}{\sqrt{|\mr|}}\right)
\end{equation}
as $|\mr|\longrightarrow\infty$ uniformly on $\hat{\mr}=\frac{\mr}{|\mr|}$.

\subsection{Introduction to subspace migration}
The imaging algorithm introduced in \cite{P1} used the structure of a singular vector of the Multi-Static Response (MSR) matrix $\mathbb{M}=(F_{pq})=(F(\hat{\mr}_p,\md_q))_{p,q=1}^{N}$, whose elements $F(\hat{\mr}_p,\md_q)$ is (\ref{FarField}) with observation number $p$ and incident number $q$. Note that by combining (\ref{ScatteredField}), (\ref{FarField}), and the asymptotic behavior of the Hankel function, the far-field pattern $F(\hat{\mr}_p,\md_q)$ can be represented as the asymptotic expansion formula (see \cite{AK} for instance)
\begin{multline}\label{AFS}
  F(\hat{\mr}_p,\md_q)\approx\rho^2\frac{\omega^2(1+j)}{4\sqrt{\omega\pi}}\sum_{m=1}^{M} \left(\frac{\eps-\eps_0}{\sqrt{\eps_0\mu_0}}\mbox{area}(\mathcal{D}_m)-\hat{\mr}_p\cdot\mathbb{T}(\mr_m)\cdot\md_q\right)
  \\\times\exp\bigg(jk_0(\md_q-\hat{\mr}_p)\cdot \mr_m\bigg).
\end{multline}
For the sake of simplicity, we eliminate the constant $\frac{\omega^2(1+j)}{4\sqrt{\omega\pi}}$ in (\ref{AFS}). Then, the incident and observation direction configurations are kept same, i.e., for each $\hat{\mr}_p=-\md_p$, the $pq$-th element of the MSR matrix $\mathbb{M}$ is given by
\begin{multline*}
  F_{pq}=F(\hat{\mr}_p,\md_q)\bigg|_{\hat{\mr}_p=-\md_p}\approx\rho^2\sum_{m=1}^M\bigg[\frac{\eps_m-\eps_0}{\sqrt{\eps_0\mu_0}}\mbox{area}(\mathcal{D}_m) \\+\frac{2\mu_0}{\mu_m+\mu_0}\mbox{area}(\mathcal{D}_m)\md_p\cdot\md_q\bigg]\exp\bigg(jk_0(\md_p+\md_q)\cdot\mr_m\bigg).
\end{multline*}

Based on the above representation of $F_{pq}$, we introduce a vector $\mathbf{D}(\mr;\omega)\in\mathbb{C}^{N\times3}$ as
\begin{equation}\label{VecD}
  \mathbf{D}(\mr;\omega):=\left(
                            \begin{array}{c}
                              \me_1\exp(jk_0\md_1\cdot\mr) \\
                              \me_2\exp(jk_0\md_2\cdot\mr) \\
                              \vdots \\
                              \me_N\exp(jk_0\md_N\cdot\mr) \\
                            \end{array}
                          \right),\quad\mbox{where}\quad\me_p=(1,\md_p)^T.
\end{equation}
Then $\mathbb{M}$ can be decomposed as follows:
\[\mathbb{M}=\sum_{m=1}^{M}\mathbf{D}(\mr_m;\omega)
\left(\begin{array}{cc}
    \displaystyle\rho^2\frac{\eps_m-\eps_0}{\sqrt{\eps_0\mu_0}}\mbox{area}(\mathcal{D}_m) & \mathbb{O}_{2\times2} \\
    \mathbb{O}_{2\times1} & \rho^2\mathbb{T}(\mr_m) \\
  \end{array}\right)
\mathbf{D}(\mr_m;\omega)^T,\]
where $\mathbb{O}_{p\times q}$ demotes the $p\times q$ zero matrix. This decomposition leads us to introduce an imaging algorithm as follows. First, let us perform the Singular Value Decomposition (SVD) as follows:
\begin{equation}\label{SVD}
  \mathbb{M}=\mathbb{US\overline{V}}^T=\sum_{m=1}^{M}\sigma_m(\omega)\mathbf{U}_m(\omega)\overline{\mathbf{V}}_m(\omega)^T,
\end{equation}
where $\mathbf{U}_m$ and $\mathbf{V}_m$ are the left and right singular vectors, respectively, and $\overline{a}$ denotes the complex conjugate of $a$. Then, based on the structure of (\ref{VecD}), we define a vector $\mathbf{\hat{D}}(\mr;\omega)\in\mathbb{C}^{N\times1}$:
\begin{equation}\label{VecDhat}
  \mathbf{\hat{D}}(\mr;\omega):=\left(
                            \begin{array}{c}
                              \mc\cdot(1,\md_1)^T\exp(jk_0\md_1\cdot\mr) \\
                              \mc\cdot(1,\md_2)^T\exp(jk_0\md_2\cdot\mr) \\
                              \vdots \\
                              \mc\cdot(1,\md_N)^T\exp(jk_0\md_N\cdot\mr) \\
                            \end{array}
                          \right),\quad\mc\in\mathbb{C}^{3\times1}\backslash\{\mathbf{0}\},
\end{equation}
and corresponding unit vector
\begin{equation}\label{VecW}
  \mathbf{W}(\mr;\omega):=\frac{\mathbf{\hat{D}}(\mr;\omega)}{|\mathbf{\hat{D}}(\mr;\omega)|}.
\end{equation}
With this, we can introduce a subspace migration as follows
\begin{equation}\label{ImagingFunction}
  \mathbb{W}(\mathbf{r};\omega):=\left|\sum_{m=1}^{M}\bigg(\overline{\mathbf{W}}(\mathbf{r};\omega)\cdot\mathbf{U}_m(\omega)\bigg) \bigg(\overline{\mathbf{W}}(\mathbf{r};\omega)\cdot\overline{\mathbf{V}}_m(\omega)\bigg)\right|.
\end{equation}
Note that since the first $M$ columns of the matrices $\{\mathbf{U}_1,\mathbf{U}_2,\cdots,\mathbf{U}_M\}$ and $\{\mathbf{V}_1,\mathbf{V}_2,\cdots,\mathbf{V}_M\}$ are orthonormal, we can observe that
\begin{align*}
  &\overline{\mathbf{W}}(\mathbf{r};\omega)\cdot\mathbf{U}_m(\omega)\approx1\quad\mbox{and}\quad\overline{\mathbf{W}}(\mathbf{r};\omega)\cdot\overline{\mathbf{V}}_m(\omega)\approx1\quad\mbox{if}\quad\mr=\mr_m\\
  &\overline{\mathbf{W}}(\mathbf{r};\omega)\cdot\mathbf{U}_m(\omega)\approx0\quad\mbox{and}\quad\overline{\mathbf{W}}(\mathbf{r};\omega)\cdot\overline{\mathbf{V}}_m(\omega)\approx0\quad\mbox{if}\quad\mr\ne\mr_m,
\end{align*}
for $m=1,2,\cdots,M$. Therefore, $\mathbb{W}(\mathbf{r};\omega)$ will plots peaks of magnitude of $1$ at $\mr=\mr_m\in\Sigma_m$, and of small magnitude at $\mr\notin\Sigma_m$ (see \cite{AGKPS,P1,P2,PL2}). Complete algorithm is summarized as follows.

\begin{algorithm}
\begin{algorithmic}[1]
\Procedure{SM}{$\omega$}
\State identify permittivity $\eps_0$ and permeability $\mu_0$ of $\mathbb{R}^2$
\State given $\omega$, initialize $\mathbb{W}(\mathbf{r};\omega)$
\For{$q=1$ \textbf{to} $N$}
   \For{$p=1$ \textbf{to} $N$}
      \State collect MSR matrix data $F(\hat{\mr}_p,\md_q)\in\mathbb{M}$\Comment{see (\ref{AFS})}
   \EndFor
\EndFor
   \State perform SVD of $\mathbb{M}=\mathbb{US\overline{V}}^T$\Comment{see (\ref{SVD})}
   \State discriminate number of nonzero singular values $M$\Comment{see \cite{PL3}}
   \State choose $\{\mathbf{U}_1,\mathbf{U}_2,\cdots,\mathbf{U}_M\}$ and $\{\mathbf{V}_1,\mathbf{V}_2,\cdots,\mathbf{V}_M\}$
   \For{$\mathbf{r}\in\Omega\subset\mathbb{R}^{2}$}\Comment{$\Omega$ is a search domain}
      \State generate $\mathbf{\hat{D}}(\mr;\omega)$ and $\mathbf{W}(\mr;\omega)$\Comment{see (\ref{VecDhat}) and (\ref{VecW})}
      \State initialize $I(\mathbf{r})$
      \For{$m=1$ \textbf{to} $M$}
        \State $I(\mathbf{r})\gets I(\mathbf{r})+
              (\overline{\mathbf{W}}(\mathbf{r};\omega)\cdot\mathbf{U}_m(\omega)) (\overline{\mathbf{W}}(\mathbf{r};\omega)\cdot\overline{\mathbf{V}}_m(\omega))$
      \EndFor
      \State $\mathbb{W}(\mathbf{r};\omega)=|I(\mathbf{r})|$
    \EndFor\Comment{see (\ref{ImagingFunction})}
\State plot $\mathbb{W}(\mathbf{r};\omega)$
\State find $\mathbf{r}=\mathbf{r}_m\in\Sigma_m$\Comment{$\mathbb{W}(\mathbf{r};\omega)\approx1$}
\EndProcedure
\end{algorithmic}
\caption{Imaging algorithm via Subspace Migration (SM)}\label{LSA}
\end{algorithm}

\section{Structure analysis of imaging function}\label{sec:3}
\subsection{Structure of imaging function (\ref{ImagingFunction})}
We now determind the structure of imaging function (\ref{ImagingFunction}). For this purpose, we recall some useful statements.
\begin{lem}[\cite{AGKPS}]\label{Lemma1}
  A relation $A\sim B$ means that there exists a constant $C$ such that $A=CB$. Then, for vectors $\mathbf{U}_m$ and $\mathbf{V}_m$ in (\ref{SVD}) and $\mathbf{W}(\mr;\omega)$ in (\ref{VecW}), the following relationship holds
  \[\mathbf{W}(\mr_m;\omega)\sim\mathbf{U}_m(\omega)\quad\mbox{and}\quad\mathbf{W}(\mr_m;\omega)\sim\overline{\mathbf{V}}_m(\omega).\]
\end{lem}

\begin{lem}[\cite{G}]\label{Lemma2}
  Let $\md,\mr\in\mathbb{R}^2$, and $\omega>0$; then
  \[\int_{\mathfrak{C}^1}\exp(j\omega\md\cdot\mr)dS(\md)=2\pi J_0(\omega|\mr|),\]
  where $J_\nu(x)$ denotes the Bessel function of order $\nu$ of the first kind.
\end{lem}

Subsequently, we can explore the structure of (\ref{ImagingFunction}) as follows
\begin{thm}\label{TheoremSingle}
  If the total number of incident and observation directions $N$ is sufficiently large and satisfies $N>M$, then the imaging function (\ref{ImagingFunction}) can be represented as follows:
  \begin{equation}\label{ImagingFunctionSingle}
    \mathbb{W}(\mathbf{r};\omega)\sim\sum_{m=1}^{M}J_0^2(\omega|\mathbf{r}_m-\mathbf{r}|).
  \end{equation}
\end{thm}
\begin{proof}
By hypothesis, we assume that $N$ is sufficiently large. For simplicity, we set $\triangle\md_p:=|\md_p-\md_{p-1}|$ for $p=2,3,\cdots,N,$ and $\triangle\md_1:=|\md_1-\md_N|$. $\triangle\md_q$ is defined analogously. Then, applying Lemmas \ref{Lemma1} and \ref{Lemma2} yields
\begin{align}
\begin{aligned}\label{BesselFunction}
  \mathbb{W}(\mathbf{r};\omega)&=\left|\sum_{m=1}^{M}\bigg(\overline{\mathbf{W}}(\mathbf{r};\omega)\cdot\mathbf{U}_m(\omega)\bigg) \bigg(\overline{\mathbf{W}}(\mathbf{r};\omega)\cdot\overline{\mathbf{V}}_m(\omega)\bigg)\right|\\
  &\sim\left|\sum_{m=1}^{M}\left(\sum_{p=1}^{N}\exp(j\omega\md_p\cdot(\mathbf{r}_m-\mathbf{r}))\frac{\triangle\md_p}{2\pi}\right) \left(\sum_{q=1}^{N}\exp(j\omega\md_q\cdot(\mathbf{r}_m-\mathbf{r}))\frac{\triangle\md_q}{2\pi}\right)\right|\\
  &\approx\frac{1}{4\pi^2}\left|\sum_{m=1}^{M}\left(\int_{\mathfrak{C}^1}\exp(j\omega\md\cdot(\mathbf{r}_m-\mathbf{r}))dS(\md)\right)^2\right| =\sum_{m=1}^{M}J_0^2(\omega|\mathbf{r}_m-\mathbf{r}|).
\end{aligned}
\end{align}
This completes the proof.
\end{proof}

Note that $J_0(x)$ has the maximum value $1$ at $x=0$. This is the reason why the map of $\mathbb{W}(\mathbf{r};\omega)$ plots magnitude $1$ at $\mathbf{r}=\mathbf{r}_m\in\Sigma_m$. Moreover, due to the oscillating property of $J_0(x)$, Theorem \ref{TheoremSingle} indicates why imaging function (\ref{ImagingFunction}) plots unexpected replicas, as shown in Figure \ref{FigureBessel}.

\subsection{Reason behind enhancement in the imaging performance by applying multiple frequencies}
According to the Theorem \ref{TheoremSingle}, the oscillating pattern of the Bessel function must be reduced or eliminated in order to improve the imaging performance. One way to do so is to apply the high-frequency $\omega=+\infty$ in theory. Another way is to apply several frequencies to the imaging function (\ref{ImagingFunction}) as follows:
\begin{equation}\label{ImagingFunctionMultiple}
  \mathbb{W}(\mathbf{r};S):=\frac{1}{S}\left|\sum_{s=1}^{S}\mathbb{W}(\mathbf{r};\omega_s)\right| =\frac{1}{S}\left|\sum_{s=1}^{S}\sum_{m=1}^{M}\bigg(\overline{\mathbf{W}}(\mathbf{r};\omega_s)\cdot\mathbf{U}_m(\omega_s)\bigg) \bigg(\overline{\mathbf{W}}(\mathbf{r};\omega_s)\cdot\overline{\mathbf{V}}_m(\omega_s)\bigg)\right|.
\end{equation}

Several researches in \cite{A,AGKPS,P1,P2} have confirmed on the basis of Statistical Hypothesis Testing and numerical experiments that the multi-frequency imaging function (\ref{ImagingFunctionMultiple}) is an improved version of the single-frequency version (\ref{ImagingFunctionSingle}). The reason for this is discussed as follows.

\begin{thm}
  If $\omega_S$ and the total number of incident and observation directions $N$ is sufficiently large and satisfies $N>M$, then the structure of the imaging function (\ref{ImagingFunctionMultiple}) is
  \begin{multline}\label{StructureImagingFunctionMultiple}
    \mathbb{W}(\mathbf{r};S)\sim\left|\sum_{m=1}^{M}\frac{\omega_S}{\omega_S-\omega_1}\bigg(J_0^2(\omega_S|\mathbf{r}_m-\mathbf{r}|)+J_1^2(\omega_S|\mathbf{r}_m-\mathbf{r}|)\bigg)\right.\\ \left.-\frac{\omega_1}{\omega_S-\omega_1}\bigg(J_0^2(\omega_1|\mathbf{r}_m-\mathbf{r}|)+J_1^2(\omega_1|\mathbf{r}_m-\mathbf{r}|)\bigg)\right|.
  \end{multline}
\end{thm}
\begin{proof}
According to (\ref{BesselFunction}), we can observe that
\[\mathbb{W}(\mathbf{r};S)\approx\frac{1}{S}\left|\sum_{s=1}^{S}\sum_{m=1}^{M}J_0^2(\omega_s|\mathbf{r}_m-\mathbf{r}|)\right| \approx\frac{1}{\omega_S-\omega_1}\left|\sum_{m=1}^{M}\int_{\omega_1}^{\omega_S}J_0^2(\omega|\mathbf{r}_m-\mathbf{r}|)d\omega\right|.\]
Using this, we apply an indefinite integral formula of the Bessel function (see \cite[page 35]{R}):
\[\int J_0^2(x)=x\bigg(J_0^2(x)+J_1^2(x)\bigg)+\int J_1^2(x)dx\]
in addition to a change of variable $\omega|\mathbf{r}_m-\mathbf{r}|=x$. This yields
\begin{align*}
  \int_{\omega_1}^{\omega_S}J_0^2(\omega|\mathbf{r}_m-\mathbf{r}|)d\omega =&\frac{1}{|\mathbf{r}_m-\mathbf{r}|}\int_{\omega_1|\mathbf{r}_m-\mathbf{r}|}^{\omega_S|\mathbf{r}_m-\mathbf{r}|}J_0^2(x)dx\\
  =&\omega_S\bigg(J_0^2(\omega_S|\mathbf{r}_m-\mathbf{r}|)+J_1^2(\omega_S|\mathbf{r}_m-\mathbf{r}|)\bigg)\\
  &-\omega_1\bigg(J_0^2(\omega_1|\mathbf{r}_m-\mathbf{r}|)+J_1^2(\omega_1|\mathbf{r}_m-\mathbf{r}|)\bigg)
  +\int_{\omega_1}^{\omega_S}J_1^2(\omega|\mathbf{r}_m-\mathbf{r}|)d\omega.
\end{align*}
Now, we consider the upper bound of
\[\Lambda(|\mathbf{r}_m-\mathbf{r}|,\omega_1,\omega_S):=\int_{\omega_1}^{\omega_S}J_1^2(\omega|\mathbf{r}_m-\mathbf{r}|)d\omega.\]
Note that since $J_1(0)=0$, let us assume that $|\mr_m-\mr|\ne0$ and $0<\omega_S|\mr_m-\mr|\ll\sqrt{2}$. Then applying asymptotic behavior
\[J_\nu(x)\approx\frac{1}{\Gamma(\nu+1)}\bigg(\frac{x}{2}\bigg)^\nu\]
and boundedness property $J_\nu(\omega|\mathbf{r}_m-\mathbf{r}|)\leq\frac{1}{\sqrt{2}}$ yields
\begin{align*}
  \int_{\omega_1}^{\omega_S}J_1^2(\omega|\mathbf{r}_m-\mathbf{r}|)d\omega
  &\leq \frac{1}{\sqrt{2}}\int_{\omega_1}^{\omega_S}J_1(\omega|\mathbf{r}_m-\mathbf{r}|)d\omega
  =\frac{1}{\sqrt{2}}\int_{\omega_1|\mathbf{r}_m-\mathbf{r}|}^{\omega_S|\mathbf{r}_m-\mathbf{r}|}\frac{J_1(x)}{|\mathbf{r}_m-\mathbf{r}|}dx\\
  &=\frac{1}{2\sqrt{2}|\mathbf{r}_m-\mathbf{r}|}\int_{\omega_1|\mathbf{r}_m-\mathbf{r}|}^{\omega_S|\mathbf{r}_m-\mathbf{r}|}xdx =\frac{(\omega_S)^2-(\omega_1)^2}{4\sqrt{2}}|\mathbf{r}_m-\mathbf{r}|\\
  &<\frac{\omega_S}{4\sqrt{2}}\bigg(\omega_S|\mathbf{r}_m-\mathbf{r}|\bigg)\ll\frac{\omega_S}{4}=O(\omega_S).
\end{align*}

Now, assume that $\omega_S$ satisfies
\[\omega_S|\mathbf{r}_m-\mathbf{r}|\gg\sqrt{2}\quad\mbox{i.e.,}\quad|\mathbf{r}_m-\mathbf{r}|\gg\frac{\sqrt{2}}{\omega_S}>0.\]
Then since
\[\int J_1(x)dx=-J_0(x),\]
we can obtain
\begin{align*}
  \int_{\omega_1}^{\omega_S}J_1(\omega|\mathbf{r}_m-\mathbf{r}|)d\omega &=\int_{\omega_1|\mathbf{r}_m-\mathbf{r}|}^{\omega_S|\mathbf{r}_m-\mathbf{r}|}\frac{J_1(x)}{|\mathbf{r}_m-\mathbf{r}|}dx
  =\frac{1}{|\mathbf{r}_m-\mathbf{r}|}\bigg[-J_0(x)\bigg]_{\omega_1|\mathbf{r}_m-\mathbf{r}|}^{\omega_S|\mathbf{r}_m-\mathbf{r}|}\\
  &=\frac{1}{|\mathbf{r}_m-\mathbf{r}|}\bigg(J_0(\omega_1|\mathbf{r}_m-\mathbf{r}|)-J_0(\omega_S|\mathbf{r}_m-\mathbf{r}|)\bigg)\\
  &\leq\frac{2}{|\mathbf{r}_m-\mathbf{r}|}\ll\sqrt{2}\omega_S.
\end{align*}
Therefore, the term $\Lambda(|\mathbf{r}_m-\mathbf{r}|,\omega_1,\omega_S)$ can be disregarded because
\begin{multline*}
  \left|\frac{\omega_S}{\omega_S-\omega_1}\bigg(J_0^2(\omega_S|\mathbf{r}_m-\mathbf{r}|)+J_1^2(\omega_S|\mathbf{r}_m-\mathbf{r}|)\bigg)\right.\\
  \left.-\frac{\omega_1}{\omega_S-\omega_1}\bigg(J_0^2(\omega_1|\mathbf{r}_m-\mathbf{r}|)+J_1^2(\omega_1|\mathbf{r}_m-\mathbf{r}|)\bigg)\right|=O(\omega_S)
\end{multline*}
and $\Lambda(|\mathbf{r}_m-\mathbf{r}|,\omega_1,\omega_S)\ll O(\omega_S)$. Hence we can obtain (\ref{StructureImagingFunctionMultiple}). This completes the proof.
\end{proof}

Two-dimensional plot for (\ref{StructureImagingFunctionMultiple}) is shown in Figure \ref{FigureBessel}. This shows that (\ref{ImagingFunctionMultiple}) yields better images owing to less oscillation than (\ref{ImagingFunctionSingle}) does. This result indicates why a multi-frequency based imaging function offers better results than a single-frequency based one.

\begin{figure}
\begin{center}
\subfigure[$\omega=\frac{2\pi}{0.3}$]{\includegraphics[width=0.49\textwidth]{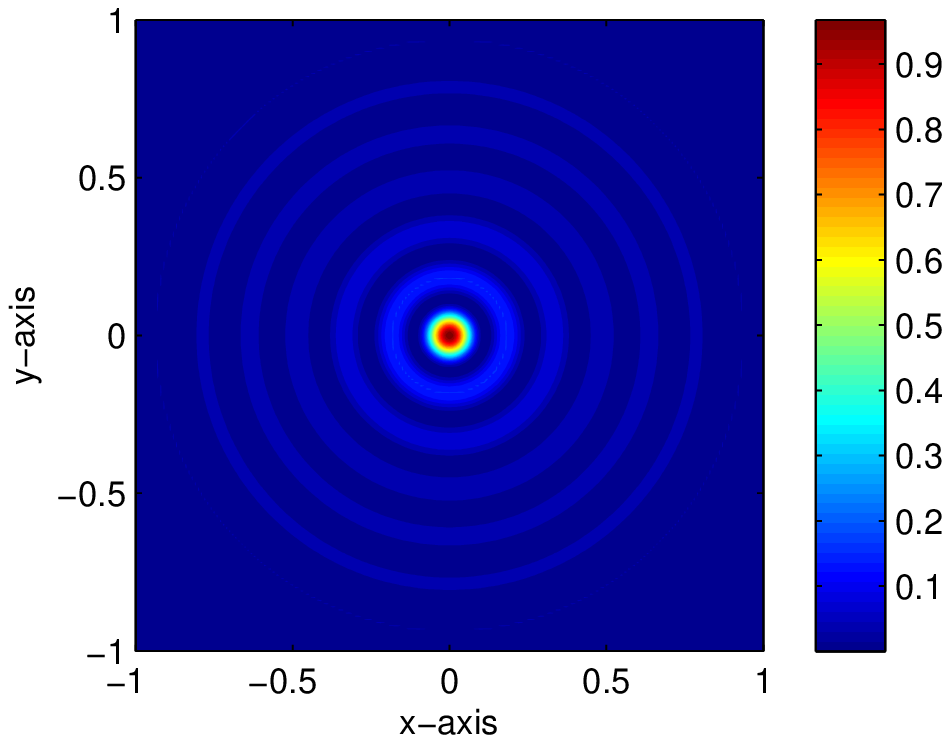}}
\subfigure[]{\includegraphics[width=0.49\textwidth]{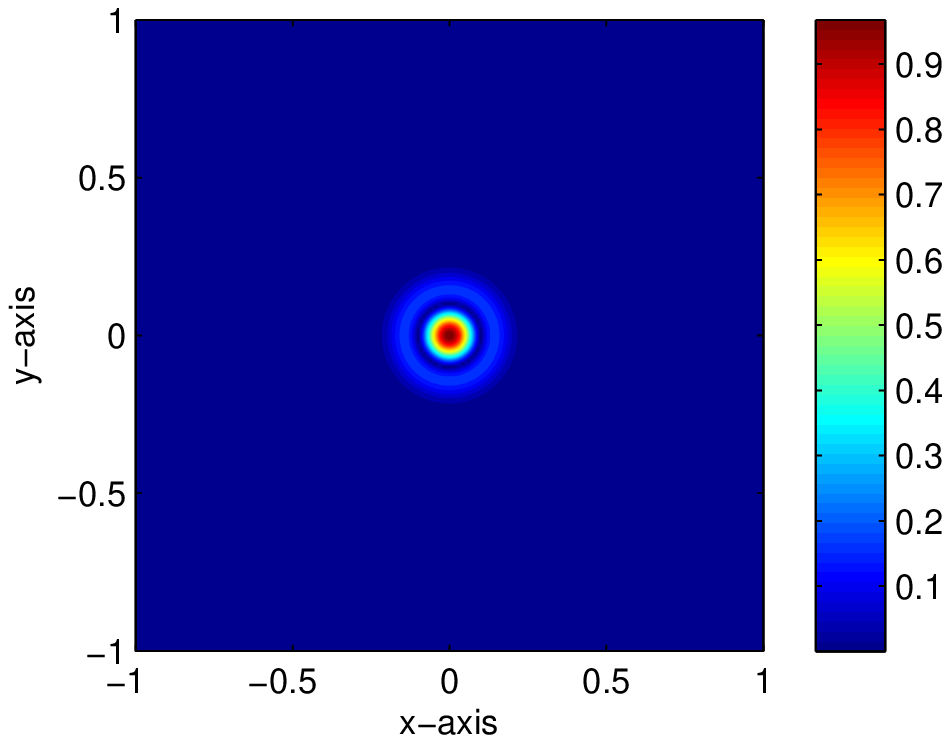}}
\end{center}
\caption{(a) 2-D plot of (\ref{ImagingFunctionSingle}) and (b) 2-D plot of (\ref{ImagingFunctionMultiple}) when $m=1$ and $\mr_m=\mathbf{0}$.}
\label{FigureBessel}
\end{figure}

\section{Numerical experiments and discussions}\label{sec:4}
In this section, we describe the numerical experiments we conducted to validate our analysis. For this purpose, we choose a set of three different small disks $\Sigma_m$. The common radii $\rho$ of $\Sigma_m$ are set to $0.1$, and parameters $\eps_0$ and $\mu_0$ are chosen as $1$. Locations $\mr_m$ of $\Sigma_m$ are selected as $\mr_1=(0.4,0)$, $\mr_2=(-0.6,0.3)$, and $\mr_3=(0.1,-0.5)$. For a given wavelength $\lambda_s$, each frequency is selected as $\omega_s=\frac{2\pi}{\lambda_s}$, for $s=1,2,\cdots,S$. Note that the test vector $\mc$ in (\ref{VecDhat}) is selected as $\mc=(5,1,1)^T$ and all the wavelengths $\lambda_s$ are uniformly distributed in the interval $[\lambda_1,\lambda_S]$. The observation directions $\md_p$ are selected as
\[\md_p=\left(\cos\frac{2\pi p}{N},\sin\frac{2\pi p}{N}\right)\quad\mbox{for}\quad p=1,2,\cdots,N,\]
and the incident directions $\md_q\in\mathfrak{C}^1$ are selected analogously.

In all the examples, scattered field data computed within the framework of the Foldy-Lax equation \cite{TKDA}. Then, a white Gaussian noise with $10$ dB signal-to-noise ratio (SNR) is added to the unperturbed data in order to exhibit the robustness of the proposed algorithm via the MATLAB command \textit{awgn}. In order to obtain the number of nonzero singular values $M$ for each frequency $\omega_s$, a $0.01$-threshold scheme is adopted (see \cite{PL1,PL3}). The search domain $\Omega\subset\mathbb{R}^2$ is selected as a square $\Omega=[-1,1]\times[-1,1]$.

Figure \ref{Result1} shows the map of $\mathbb{W}(\mathbf{r};10)$ via the MSR matrix $\mathbb{M}$ for $N=20$ and $S=10$ and different frequencies with $\lambda_1=0.5$ and $\lambda_S=0.3$. On the left-hand side of Figure \ref{Result1}, we set the same material properties $\eps_m\equiv5$ and $\mu_m\equiv5$, $m=1,2,3$. As expected, locations of $\Sigma_m$ can be clearly identified. On the right-hand side of Figure \ref{Result1}, we set different material properties $\eps_1=\mu_1=5$, $\eps_2=\mu_2=2$, and $\eps_3=\mu_3=7$. Note that due to the small values of $\eps_2$ and $\mu_2$, the map of $\mathbb{W}(\mathbf{r};10)$ plots a small magnitude at $\mr_2\in\Sigma_2$ but the locations of all $\Sigma_m$ are well identified.

\begin{figure}
\begin{center}
\subfigure[same material property]{\includegraphics[width=0.49\textwidth]{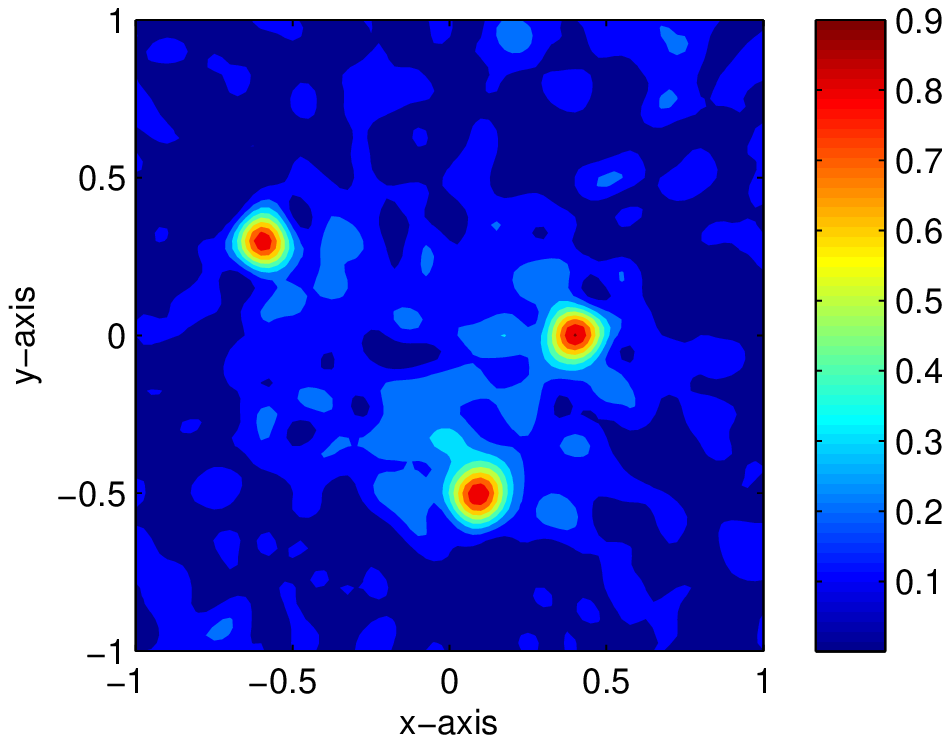}}
\subfigure[different material property]{\includegraphics[width=0.49\textwidth]{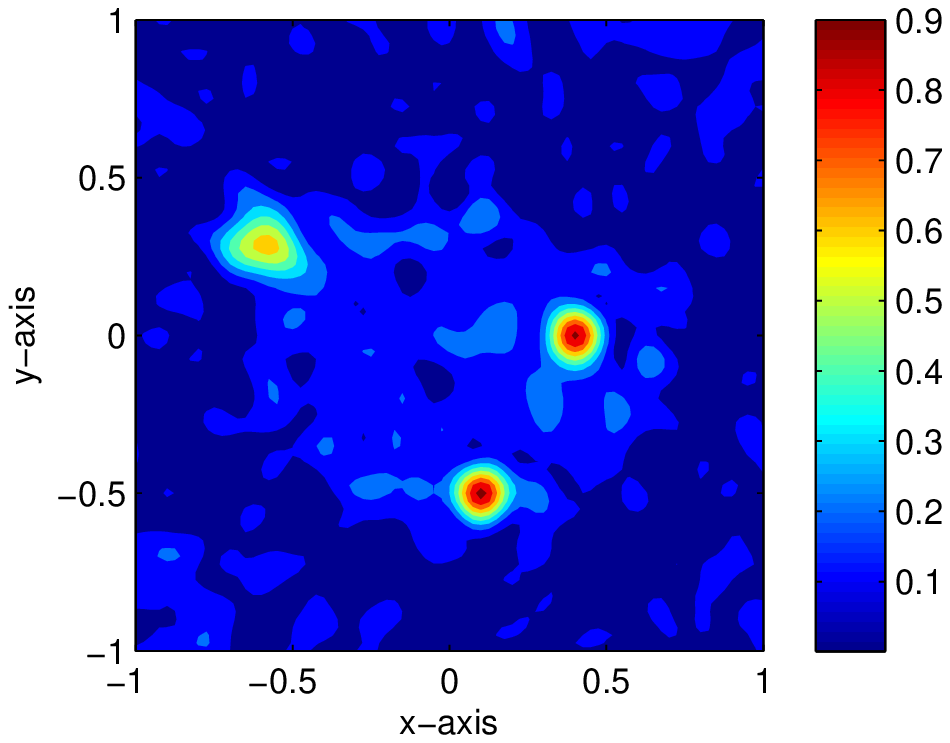}}
\end{center}
\caption{Maps of $\mathbb{W}(\mathbf{r};10)$.}
\label{Result1}
\end{figure}

Figure \ref{InfluenceFrequencies} shows the influence of the number of applied frequencies $S$. As we discussed in section \ref{sec:3}, increasing $S$ yields a more accurate image. Note that applying an infinite number of $S$ would yield good results in theory, but in this experiment, $S=10$ is sufficient for obtaining a good result.

\begin{figure}
\begin{center}
\subfigure[$S=1$]{\includegraphics[width=0.49\textwidth]{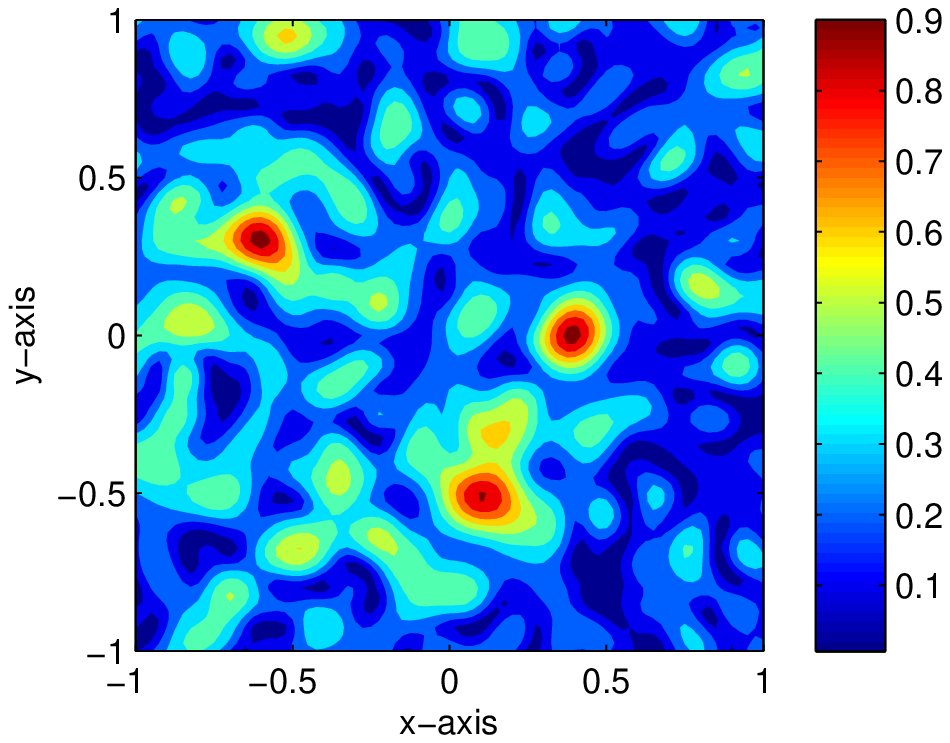}}
\subfigure[$S=3$]{\includegraphics[width=0.49\textwidth]{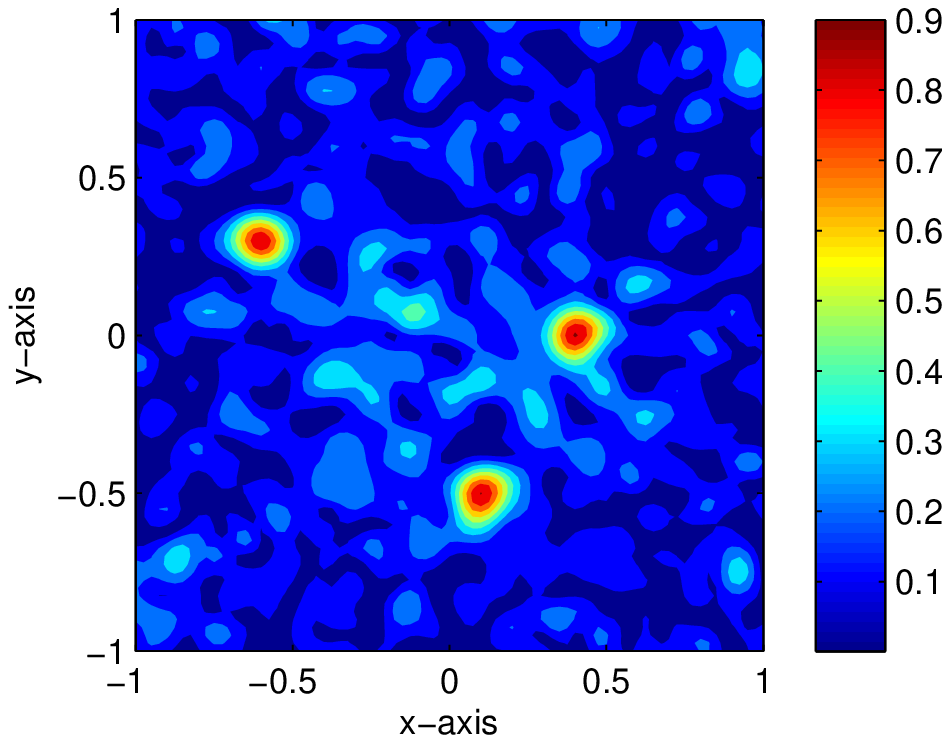}}
\subfigure[$S=7$]{\includegraphics[width=0.49\textwidth]{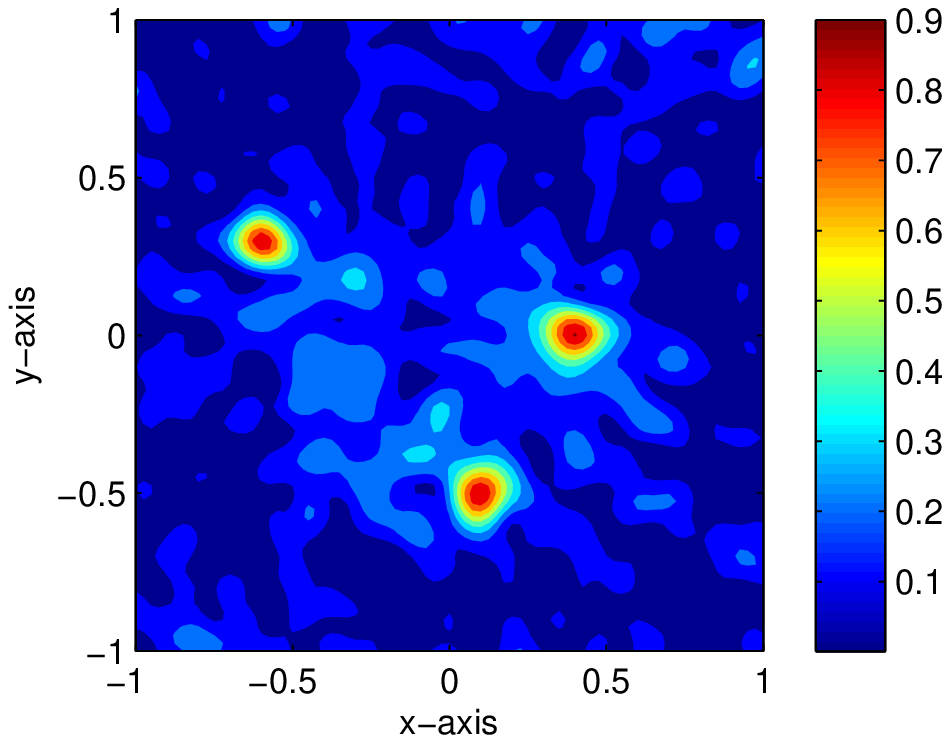}}
\subfigure[$S=20$]{\includegraphics[width=0.49\textwidth]{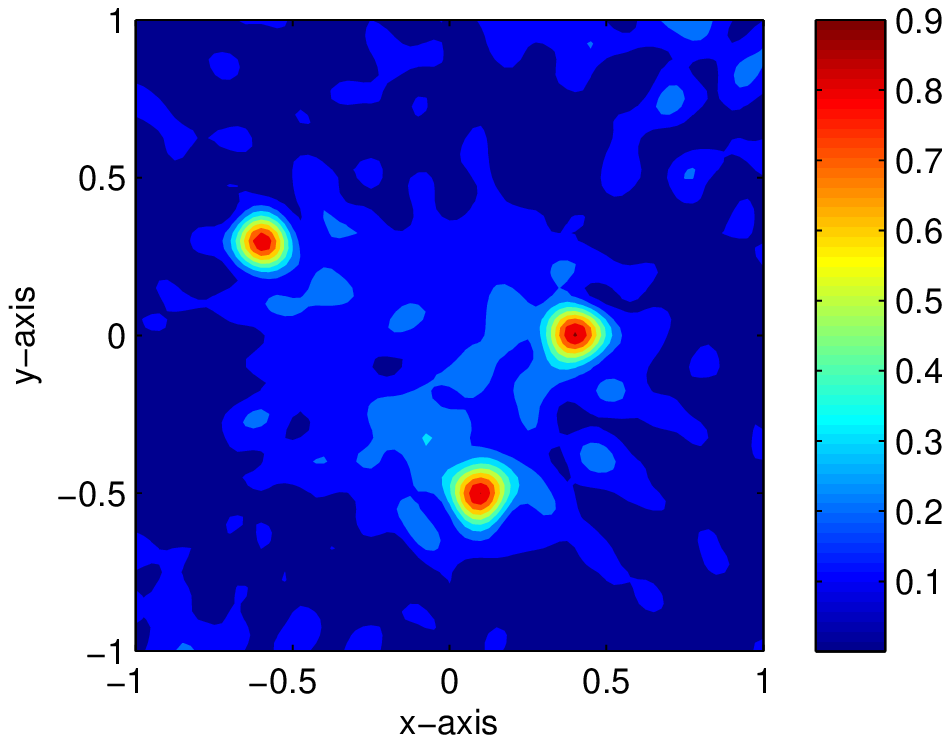}}
\end{center}
\caption{Maps of $\mathbb{W}(\mathbf{r};S)$.}
\label{InfluenceFrequencies}
\end{figure}

Based on the recent work \cite{P1}, the proposed algorithm can be applied to imaging of extended, crack-like electromagnetic inclusion(s) $\Gamma$ with a supporting curve $\gamma$ and a small thickness $h$. However, note that even a sufficiently large number of $N$ and $S$ applied to obtain a good image of a complex-shaped thin inclusion using the proposed algorithm can occasionally yield poor results (see Figure \ref{ComplexInclusion}). Furthermore, note that the elements of $\mathbb{M}$ are expressed as written by
\begin{multline*}
  F(\md_p,\md_q)\sim\sum_{m=1}^{M}\left[\frac{\eps-\eps_0}{\sqrt{\eps_0\mu_0}}
  +2\bigg(\frac{1}{\mu}-\frac{1}{\mu_0}\bigg)\md_p\cdot\mt(\mr_m)\md_q\cdot\mt(\mr_m)\right.\\
  \left.+2\bigg(\frac{1}{\mu_0}-\frac{\mu}{\mu_0^2}\bigg)\md_p\cdot\mn(\mr_m)\md_q\cdot\mn(\mr_m)\right]\exp\bigg(jk_0(\md_p+\md_q)\cdot \mr_m\bigg).
\end{multline*}
Therefore, $\mc$ in (\ref{VecDhat}) must be a linear combination of a unit tangential vector $\mt(\mr_m)$ and a normal vector $\mn(\mr_m)$ at $\mr_m\in\gamma$. If we have \textit{a priori} information of $\mt(\mr_m)$ and $\mn(\mr_m)$, we can obtain a good result. However, because this is not the case, it is difficult to obtain a good result. This is further explained in detail in \cite[Section 4.3.1]{PL3}.

\begin{figure}
\begin{center}
\subfigure[$N=48$ and $S=10$]{\includegraphics[width=0.49\textwidth]{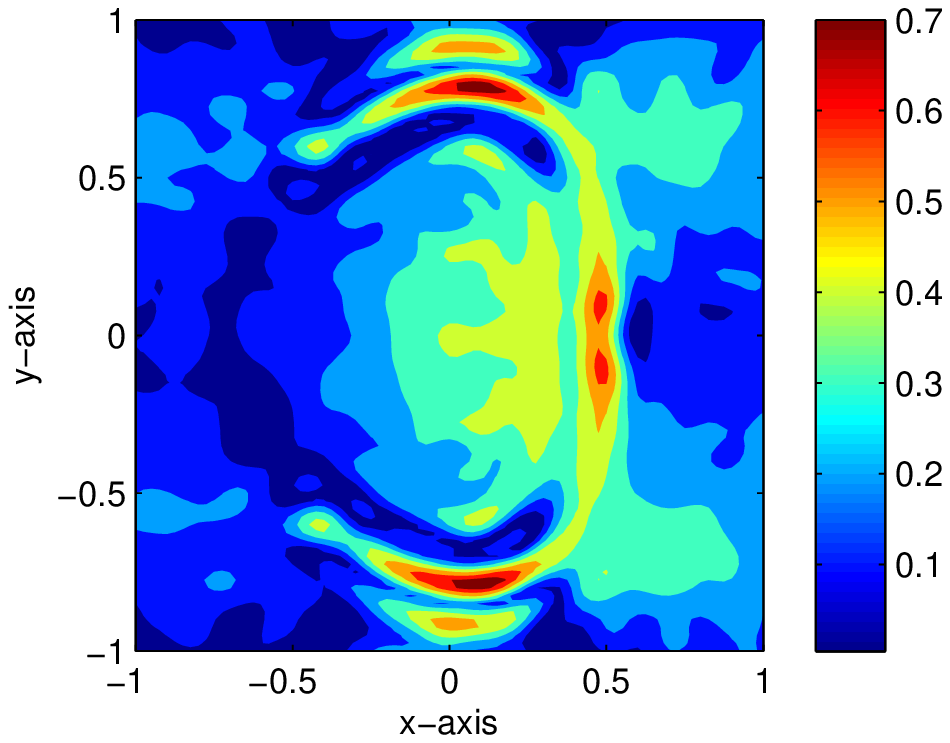}}
\subfigure[$N=64$ and $S=24$]{\includegraphics[width=0.49\textwidth]{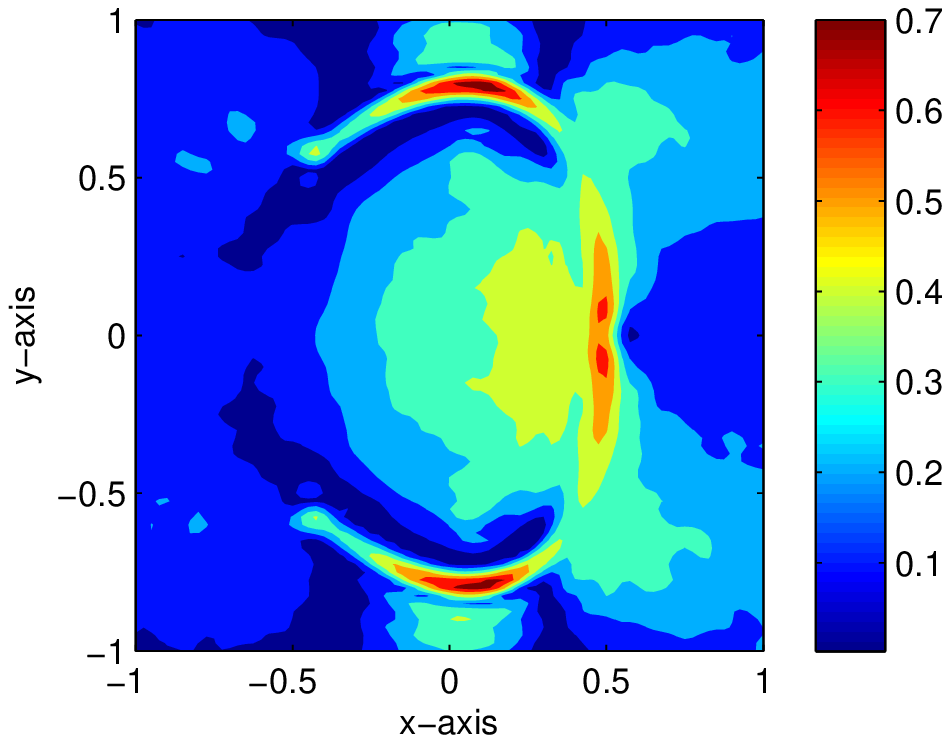}}
\end{center}
\caption{Maps of $\mathbb{W}(\mathbf{r};S)$.}
\label{ComplexInclusion}
\end{figure}

\section{Conclusion}\label{sec:5}
Using an integral representation formula and an indefinite integral of the Bessel function, we determined the structure of single and multiple electromagnetic imaging functions. Because of the oscillation aspect of the Bessel function, we confirmed the reason behind the improved imaging performance by successfully applying high and multiple frequencies.

Based on recent works in \cite{AGKPS,AKLP,P1,P2,PL2}, it has shown that subspace migration offers better results than the MUSIC and Kirchhoff migrations. Specially, subspace migration can be applied to limited-view inverse scattering problems. However, in order to determine the structure of subspace migration in the limited-view problem, the integration in Lemma \ref{Lemma2} on the subset of a unit circle must be evaluated; however, this evaluation is very difficult to perform. Therefore, identifying the imaging function structure in the limited-view problem will prove to be an interesting research topic. Moreover, in the imaging of crack-like inclusions, estimating unit tangential and normal vectors on such an inclusion and yielding relatively good results will be an interesting work.

Finally, we considered the imaging function for penetrable electromagnetic inclusions but it will be applied to the perfectly conducting inclusion(s) directly. Extension to the perfectly conducting target will be a forthcoming work. We further believe that the proposed strategy can be extended to a three-dimensional problem.

\section*{acknowledgments}
W.-K. Park would like to thank Habib Ammari for introducing \cite{A} and many valuable advices. This work was supported by the research program of Kookmin University in Korea, the Basic Science Research Program through the National Research Foundation of Korea (NRF) funded by the Ministry of Education, Science and Technology (No. 2012-0003207), and the WCU(World Class University) program through the National Research Foundation of Korea(NRF) funded by the Ministry of Education, Science and Technology R31-10049.


\begin{thebibliography}{00}
\bibitem{AS} M. Abramowitz and I. A. Stegun, {\em Handbook of Mathematical Functions, with Formulas, Graphs, and Mathematical Tables} (Dover, New York, 1996).
\bibitem{ADIM}D. \'Alvarez, O. Dorn, N. Irishina, and M. Moscoso, J. Comput. Phys., {\bf 228}, 5710 (2009).
\bibitem{A}H. Ammari, {\em Mathematical Modeling in Biomedical Imaging II: Optical, Ultrasound, and Opto-Acoustic Tomographies}, Lecture Notes in Mathematics, {\bf 2035} (Springer-Verlag, Berlin, 2011).
\bibitem{AGJK} H. Ammari, J. Garnier, V. Jugnon, and H. Kang, SIAM J. Control. Optim., {\bf 50}, 48 (2012).
\bibitem{AGKPS}H. Ammari, J. Garnier, H. Kang, W.-K. Park, and K. S{\o}lna, SIAM J. Appl. Math., {\bf 71}, 68 (2011).
\bibitem{AK}H. Ammari and H. Kang, {\em Reconstruction of Small Inhomogeneities from Boundary Measurements, Lecture Notes in Mathematics}, {\bf 1846} (Springer-Verlag, Berlin, 2004).
\bibitem{AKLP}H. Ammari, H. Kang, H. Lee and W.-K. Park, SIAM J. Sci. Comput., {\bf 32}, 894 (2010).
\bibitem{C}X. Chen, J. Electromagn. Waves Appl., {\bf 23} 1397 (2009).
\bibitem{CH}M. Cheney, Inverse Problems, {\bf 17}, 591 (2001).
\bibitem{CHM}D. Colton, H. Haddar and P. Monk, SIAM J. Sci. Comput., {\bf 24}, 719 (2002).
\bibitem{DEKPS}F. Delbary, K. Erhard, R. Kress, R. Potthast and J. Schulz, Inverse Problems, {\bf 24}, 015002 (2008).
\bibitem{D}M. Donelli, Prog. Electromagn. Res. M, {\bf 19}, 173 (2011).
\bibitem{DCGS}M. Donelli, I. J. Craddock, D. Gibbins and M. Sarafianou, Prog. Electromagn. Res. M, {\bf 18}, 179 (2011).
\bibitem{DL}O. Dorn and D. Lesselier, Inverse Problems, {\bf 22}, R67 (2006).
\bibitem{G} R. Griesmaier, Inverse Problems, {\bf 27}, 085005 (2011).
\bibitem{KR}A. Kirsch and S. Ritter, Inverse Problems, {\bf 16}, 89 (2000).
\bibitem{KSY}O. Kwon, J. K. Seo, and J.-R. Yoon, Commun. Pur. Appl. Math., {\bf 55}, 1 (2002).
\bibitem{LB}D. Lesselier and Duchene B., Prog. Electromagn. Res., {\bf 5}, 351 (1991).
\bibitem{MKP}Y.-K. Ma, P.-S. Kim and W.-K. Park, Prog. Electromagn. Res., {\bf 122} 311 (2012).
\bibitem{P1}W.-K. Park, Prog. Electromagn. Res., {\bf 106}, 225 (2010).
\bibitem{P2}W.-K. Park, Inverse Problems, {\bf 26}, 074008 (2010).
\bibitem{P3}W.-K. Park, Prog. Electromagn. Res., {\bf 110}, 237 (2010).
\bibitem{P4}W.-K. Park, J. Comput. Phys., {\bf 231}, 1426 (2012).
\bibitem{PL1}W.-K. Park and D. Lesselier, J. Comput. Phys., {\bf 228}, 8093 (2009).
\bibitem{PL2}W.-K. Park and D. Lesselier, Waves Random Complex Media, {\bf 22}, 3 (2012).
\bibitem{PL3}W.-K. Park and D. Lesselier, Inverse Problems, {\bf 25}, 075002 (2009).
\bibitem{PL4}W.-K. Park and D. Lesselier, Inverse Problems, {\bf 25}, 085010 (2009).
\bibitem{R}W. Rosenheinrich, \url{http://www.fh-jena.de/~rsh/Forschung/Stoer/besint.pdf}.
\bibitem{TKDA}L. Tsang, J. A. Kong, K.-H. Ding, and C. O. Ao, {\em Scattering of Electromagnetic Waves: Numerical Simulations} (New York, Wiley, 2001).
\end{thebibliography}
\end{document}